\documentclass[letterpaper, 10 pt, conference]{ieeeconf}

  \usepackage{bbding}

\IEEEoverridecommandlockouts                              

\usepackage{graphics}
\usepackage{} 
\usepackage{epsfig} 
\usepackage{mathptmx} 
\usepackage{times} 
\usepackage{amsmath} 
\usepackage{amssymb}  
\usepackage{verbatim} 
\usepackage[linesnumbered,ruled,vlined]{algorithm2e}
\usepackage{subfigure}
\usepackage{url,cite}

\usepackage{color}

\newtheorem{theorem}{Theorem}
\newtheorem{lemma}{Lemma}
\newtheorem{definition}{Definition}
\newtheorem{remark}{Remark}

\newtheorem{assumption}{Assumption}
\newcommand{\Tr}{\textbf{\textrm{tr}}}
\newcommand{\Id}{\textbf{\textrm{I}}}

\newcommand{\Bcal}{\mathcal{B}}\newcommand{\Ecal}{\mathcal{E}}\newcommand{\Gcal}{\mathcal{G}}\newcommand{\Rcal}{\mathcal{R}}\newcommand{\Vcal}{\mathcal{V}}
\newcommand{\Dfrak}{\mathfrak{D}}\newcommand{\Nfrak}{\mathfrak{N}}
\newcommand{\Emb}{\mathbb{E}}\newcommand{\Rmb}{\mathbb{R}}\newcommand{\Zmb}{\mathbb{Z}}

\newcommand{\xih}{\hat{x}_{i,k|k}}\newcommand{\xib}{\hat{x}_{i,k|k-1}}
\newcommand{\eih}{\hat{\eta}_{i,k|k}}\newcommand{\eib}{\hat{\eta}_{i,k|k-1}}
\newcommand{\Pih}{P_{i,k|k}}\newcommand{\Pib}{P_{i,k|k-1}}
\newcommand{\xjh}{\hat{x}_{j,k|k}}\newcommand{\xjb}{\hat{x}_{j,k|k-1}}
\newcommand{\ejb}{\hat{\eta}_{j,k|k-1}}
\newcommand{\Pjh}{P_{j,k|k}}\newcommand{\Pjb}{P_{j,k|k-1}}
\newcommand{\onebf}{\textbf{1}}
\newcommand{\zerobf}{\textbf{0}}

\title{\LARGE \bf
On the Consistency and Confidence of Distributed Dynamic State Estimation in Wireless Sensor Networks
}

\author{Shaocheng Wang and Wei Ren
\thanks{Shaocheng Wang and Wei Ren are with the University of California, Riverside, CA 92507, USA. Email: shaocheng.wang@email.ucr.edu, ren@ece.ucr.edu.
{This paper has been accepted by the 54th IEEE Conference on Decision and Control in December, 2015 at Osaka, Japan.}
}
}

\begin{document}

\maketitle
\thispagestyle{empty}
\pagestyle{empty}


\begin{abstract}
The problem of distributed dynamic state estimation in wireless sensor networks is studied.
Two important properties of local estimates, namely, the consistency and confidence, are emphasized.
On one hand, the consistency, which means that the approximated error covariance is lower bounded by the true unknown one, has to be guaranteed so that the estimate is not over-confident.
On the other hand, since the confidence indicates the accuracy of the estimate, the estimate should be as confident as possible.
We first analyze two different information fusion strategies used in the case of information sources with, respectively, uncorrelated errors and unknown but correlated errors.
Then a distributed hybrid information fusion algorithm is proposed,
where each agent uses the information obtained not only by itself, but also from its neighbors through communication.
The proposed algorithm not only guarantees the consistency of the estimates, but also utilizes the available information sources in a more efficient manner and hence improves the confidence.
Besides, the proposed algorithm is fully distributed and guarantees convergence with the sufficient condition formulated.
The comparisons with existing algorithms are shown.
\end{abstract}

\section{Introduction}
One of the most fundamental but important properties that an estimator should have is the \emph{consistency} \cite{jazwinski1970stochastic,JulierUhlmann97_CI}.
That is, the approximated error covariance of an estimator should be lower bounded by the true error covariance.
The approximated error covariance of an inconsistent estimate does not actually indicate its uncertainty \cite{drummond2006target}.
This would cause issues in performance when it is further used in the downstream functions (data associations etc.).
The consistency is especially important in the information fusion of distributed sensor networks.
When fusing the information sources with unknown correlations, simply ignoring the unknown correlations might cause inconsistency in estimation.
The \emph{Covariance Intersection} (CI, \cite{JulierUhlmann97_CI}) algorithm is proposed to guarantee the consistency of estimates in the fusion process without knowing the cross-correlations between information sources.
On the other hand, although an estimate should be consistent, it should not be too conservative.
An over-conservative estimate indicates less confidence, which further indicates that the estimate is less useful.
Therefore, it is significant to maximize the confidence while maintaining the consistency of the estimates.

\emph{Related work}.
Ref. \cite{RenBeardKingston05} considers the consensus algorithm from the Kalman filters' perspective.
The algorithm is further modified in \cite{AlighanbariUnbiasedKalman} so that unbiased estimates can still be obtained even if the outflows of each agent are not equal to each other.
Both algorithms in \cite{RenBeardKingston05} and \cite{AlighanbariUnbiasedKalman} are proposed to estimate a static state.
Ref. \cite{olfati2007distributedKalman} proposes the well-known Kalman-consensus filter (KCF), which combines the consensus filter and the local Kalman filter, and estimates the state of a linear dynamic process.
The optimality and stability of the KCF is analyzed in \cite{olfati2009kalman}.
The KCF requires the assumption that each agent and its neighbors have joint observability of the state of interest.
If a state of interest is observed by neither a certain agent nor any of its local neighbors, this agent is referred as being \emph{naive} about the state \cite{kamal2011GKCF,kamal2013ICFTAC}.
An example of camera networks is shown in \cite{kamal2011GKCF}, where the generalized KCF (GKCF) is proposed, and is shown to outperform the KCF especially when there exist naive agents.
The same situation is considered in \cite{kamal2013ICFTAC}, where the information-consensus filter (ICF) is proposed, and is shown to asymptotically approach the centralized Kalman filter (CKF) when there are infinite communication steps in between local updates.
Recently, Ref. \cite{Battistelli2014KLA_Automatica} proposes an innovative approach by achieving consensus on the local probability density functions (PDFs), referred as the ``Kullback-Leibler average" (KLA) therein.
It is analytically shown in \cite{Battistelli2014KLA_Automatica} that each agent will eventually achieve bounded approximated estimate error covariance if the graph is strongly connected, even if each agent communicates with its neighbors for only once before updating its local estimate.

\emph{Contributions. }
While the current literature primarily focuses on the recovery of centralized estimators or the estimate errors of distributed algorithms, in this paper, we emphasize on the importance of the consistency and confidence of estimates.
Considered from the perspective of information fusion,
a \emph{distributed hybrid information fusion} (DHIF) algorithm is proposed.
In the proposed DHIF, each agent cooperates with its neighbors, utilizes all available information sources, and comes up with its local estimate that is consistent while relatively confident.
We focus on a more realistic scenario where only one communication step is allowed between neighbors before they update their local estimates.
The proposed algorithm is robust against naive agents and does not require any global parameters.
{The sufficient condition for bounded approximated local estimate error covariance is formulated.}
The comparisons with existing algorithms are shown.

\section{Preliminaries}
\subsection{Notations}
Through this paper, we use the following notations.
$\onebf$ is an all-one vector with appropriate dimension.
$\zerobf$ is an all-zero matrix with appropriate dimension.
$\Id_n$ is the $n\times n$ identity matrix.
$\Emb[\cdot]$ denotes the expectation of a random variable.
$\Zmb^*$ is the set of non-negative integers.
$\Zmb^+$ is the set of positive integers.
For any arbitrary square matrices $X_1$ and $X_2$ with the same dimensions,
$X_1\succeq X_2$ (respectively, $X_1\succ X_2$) implies that $X_1-X_2$ is a positive semi-definite (respectively, positive definite) matrix.
$X_1\geq X_2$ (respectively, $X_1> X_2$) implies that $X_1-X_2$ is a non-negative (respectively, positive) matrix.
For an arbitrary matrix $X$, $\{X\}_{(i,j)}$ is the entry on its $i$th row, $j$th column.
$\Tr(X)$ is the trace of $X$.
When we refer $X\succ\zerobf$ as a large/small matrix, we mean $X$ is large/small in the sense of the mean squared error (MSE).
If $X$ is nonsingular,
$X^{-n\top}\triangleq ((X^n)^{-1})^\top$ for some $n\in\Zmb^{*}$.

\subsection{Graph Theory}

A directed graph $\mathcal{G(V,E)}$ is used to represent the communication topology of a large-scale sensor network, where $\mathcal{V}\triangleq\{1,\cdots,N\}$ and $\mathcal{E\subseteq V\times V}$ are, respectively, the set of $N$ vertices that stands for the local agents, and the set of edges that stands for the communication channels.
Several basic concepts of graph theory used later in this paper are briefly listed here.
An \emph{edge} $(i,j)\in\mathcal{E}$ denotes that agent $j$ can receive information from agent $i$.
If a graph is \emph{undirected}, $(i,j)\in\mathcal{E}$ implies that $(j,i)\in\mathcal{E}$.
A \emph{directed path} from agent ${i_0}$ to agent ${i_\ell}$ is a sequence of vertices ${i_0},{i_1},\cdots,{i_\ell}$ such that $({i_{j-1}},{i_{j}})\in\mathcal{E}$ for $0<j\leq \ell$.
A directed (respectively, undirected) graph is \emph{strongly connected} (respectively, \emph{connected}) if there exists at least one path from every vertex to every other vertex.
$N_i\triangleq\{j|(j,i)\in\mathcal{E},\ \forall j\neq i\}$ is the \emph{neighborhood} of agent $i$.
$J_i\triangleq N_i\cup\{i\}$ is the \emph{inclusive neighborhood} of agent $i$.
A graph is \emph{complete} if $j\in N_i$ for any $i$ and $j$.
The \emph{in-degree} of a certain agent $i$ is defined as $\Delta_i \triangleq |N_i|$.
The \emph{maximum in-degree} of a graph is defined as $\Delta_\textrm{max} \triangleq \max_{i\in\{1,\ldots,N\}} \Delta_i$.
A \emph{directed spanning tree} is a subgraph of $\mathcal{G(V,E)}$ such that this subgraph is a directed tree and contains all vertices of $\mathcal{G(V,E)}$.
For a graph $\mathcal{G(V,E)}$, an associated \emph{stochastic matrix} $D$ is defined as follows.
For each $i\in\Vcal$, $\{D\}_{(i,j)}>0$ if $j\in J_i$ and $\{D\}_{(i,j)}=0$ otherwise.
Moreover, $\sum_{j}\{D\}_{(i,j)}=1$.
If it also holds that $\sum_{i}\{D\}_{(i,j)}=1,\ \forall j$, then $D$ is a \emph{doubly stochastic matrix}.
\begin{definition}[Induced family of stochastic matrices]
Let $D$ be a stochastic matrix associated with the graph $\mathcal{G(V,E)}$.
Then the family of stochastic matrices induced by $D$ is the set of all possible stochastic matrices associated with $\mathcal{G(V,E)}$.
\end{definition}

\subsection{Discrete-time Average Consensus Algorithm}


An average consensus algorithm computes the global average of some variables of interest, through communication with only local neighbors.
Suppose that $\mathbf{a}_i[0]$ is the initial value of agent $i$'s local variable.
The objective is to compute $\frac{1}{N} \sum_{i=1}^N\mathbf{a}_i[0]$ in a distributed manner.
Suppose that the graph is undirected, the local algorithm implemented at agent $i$ has the form of
$\mathbf{a}_i[k+1]=\mathbf{a}_i[k] + \epsilon\sum_{j\in N_i}(\mathbf{a}_j[k] -\mathbf{a}_i[k] ), \ \forall k\in\Zmb^*,$
where $\epsilon\in(0,1/\Delta_{\textrm{max}})$ is the rate parameter.
By iteratively updating the local variable value in this manner, the global average can be asymptotically achieved if and only if the graph is connected.
More details about the average consensus algorithm can be found in \cite{SaberFaxMurray07_IEEE}.

\subsection{Linear Information Fusion}

Consider a set of $p$ pieces of information sources, denoted as $\{a_i\}^{p}_{i=1}$,
where each $a_i\in\Rmb^{m_i}$ is the addition of a quantity linearly related to the same unknown parameter of interest $\alpha\in\Rmb^n$, and some unknown noise/errors with zero mean.
That is, $\Emb[a_i]=C_i\alpha$ with $C_i\in\Rmb^{m_i\times n},\ \forall i$.
For each $a_i$, assume that the noise/error covariance, denoted as $\tilde{R}_{a_i}\triangleq\Emb[(a_i-\Emb[a_i])(a_i-\Emb[a_i])^\top]$, is positive definite and might be unknown, but a \emph{consistent} (\cite{jazwinski1970stochastic,JulierUhlmann97_CI}) approximation, denoted as ${R}_{a_i}$, is known, i.e., $R_{a_i}\succeq\tilde{R}_{a_i}\succ\zerobf$.

The objective is to linearly fuse $\{a_i\}^{p}_{i=1}$ and obtain an unbiased consistent estimate of $\alpha$, i.e., $\hat{\alpha} = \sum_{i=1}^{p} K_i a_i$,
where $K_i\in\Rmb^{n \times m_i},\ \forall i$ are weighting matrices to be selected.
Note that $\sum_{i=1}^{p} K_i C_i = \Id_n$ due to the required unbiasedness of $\hat{\alpha}$.
Let $\tilde{R}_{\hat{\alpha}}\triangleq\Emb[(\hat{\alpha}-\Emb[\hat{\alpha}])(\hat{\alpha}-\Emb[\hat{\alpha}])^\top]$
be the true error covariance of $\hat{\alpha}$.
Let $\tilde{R}_{a_{ij}}\triangleq\Emb[(a_i-\Emb[a_i])(a_j-\Emb[a_j])^\top]$.
It follows that 
\begin{equation}\label{Equ:GenFusedCov}
\tilde{R}_{\hat{\alpha}} =  \sum_{i} K_i \tilde{R}_{a_i} K_i^\top +   \sum_{i}   \sum_{j \neq i} K_i \tilde{R}_{a_{ij}} K_j^\top.
\end{equation}
When $\tilde{R}_{a_{i}}$ and $\tilde{R}_{a_{ij}}$ are unknown, an approximation of $\tilde{R}_{\hat{\alpha}}$,
is required.
This is when the problem becomes tricky.
On one hand, the consistency of the estimate has to be guaranteed since an inconsistent estimate is over-confident and the corresponding approximated error covariance does not realistically implies the uncertainty \cite{drummond2006target}.
On the other hand, as the error covariance is a measure of estimate uncertainty.
For a consistent estimate, a smaller approximated error covariance indicates more confidence.
In the rest of this paper, we assume that when fusing information sources, the objective is to minimize the trace of the fused error covariance.
Two situations are considered.

\subsubsection{Fusion of information with uncorrelated errors}

When the information sources have mutually uncorrelated errors, i.e., $\tilde{R}_{a_{ij}}=\zerobf,\ \forall i,j$, it follows from \eqref{Equ:GenFusedCov} that
$\tilde{R}_{\hat{\alpha}} =  \sum_{i} K_i \tilde{R}_{a_i} K_i^\top$.
Thus a reasonable approximation of $\tilde{R}_{\hat{\alpha}}$ could be
$R_{\textrm{u}}\triangleq\sum_{i} K_i {R}_{a_i} K_i^\top$.
Note that if $R_{a_i}\succeq\tilde{R}_{a_i},\ \forall i$, it is guaranteed that $R_{\textrm{u}}\succeq \tilde{R}_{\hat{\alpha}}$.
The optimal $R_{\textrm{u}}$ and $\hat{\alpha}$, denoted as, respectively, $R_{\textrm{u}}^*$ and $\hat{\alpha}^*$, can be obtained by minimizing $\Tr (R_{\textrm{u}})$ with $\sum_{i} K_i C_i = \Id_n$ being held to guarantee the unbiasedness.
\begin{lemma}
\label{Lemma:FusionOfIndepInfo}
Let $K_i^*$ be the optimal weighting matrices such that $\Tr (\sum_{i} K_i {R}_{a_i} K_i^\top)$ is minimized with $\sum_{i} K_i C_i = \Id_n$ satisfied.
Let $\hat{\alpha}^*\triangleq \sum_i K_i^* a_i$ and $R_{\textrm{u}}^*\triangleq K_i^* {R}_{a_i} {K_i^*}^\top$.
Then, $K_i^* = ( C_i^\top R_{a_i}^{-1} C_i)^{-1}C_i^\top R_{a_i}^{-1},\ \forall i$ and
\begin{equation}\label{Equ:FusionIndep}
R_{\textrm{u}}^* = (\sum_i C_i^\top R_{a_i}^{-1} C_i)^{-1},\quad
\hat{\alpha}^* = R_{\textrm{u}}^*(\sum_i C_i^\top R_{a_i}^{-1} a_i).
\end{equation}
\end{lemma}
\begin{proof}
Let $\Lambda \in \Rmb^{n\times n}$ be the Lagrangian multiplier.
Construct the following augmented function
\begin{equation*}
L(K_1,\cdots,K_p,\Lambda)\triangleq \Tr \sum_{i} K_i {R}_{a_i} K_i^\top + \sum_{j,k}\{\Lambda(\sum_{i} K_i C_i -\Id_n)\}_{(j,k)}.
\end{equation*}
By setting the partial derivative with respect to $\Lambda$ and each $K_i$ to zero, it follows that $K_i^*=\frac{1}{2}\Lambda^* C_i^\top R_{a_i}^{-1},\ \forall i$ and $\Lambda^* = 2(\sum_{i} C_i^\top {R}_{a_i}^{-1} C_i)^{-1}$.
Accordingly, $\hat{\alpha}^*$ and $R_{\textrm{u}}^*$ can be obtained by their definitions.
\end{proof}


In summary, the fusion steps given by \eqref{Equ:FusionIndep} preserve the consistency, and yield the optimal estimate.
Note that if ${R}_{a_i}=\tilde{R}_{a_i}, \forall i$, one can obtain $R^*_u=\tilde{R}_{\hat{\alpha}}$.
However, this only holds when the information sources have mutually uncorrelated errors.

\subsubsection{Fusion of information with unknown correlated errors}
When the information sources have unknown but correlated errors, the last term in \eqref{Equ:GenFusedCov} is nonzero and unknown.
In such a case, if one approximates $\tilde{R}_{\hat{\alpha}}$ with  $R_{\textrm{u}}$ as in the previous case, it is not guaranteed that $R_{\textrm{u}}\succeq\tilde{R}_{\hat{\alpha}}$ even if $R_{a_i}\succeq\tilde{R}_{a_i},\ \forall i$.
The {Covariance Intersection} (CI) algorithm is proposed in \cite{JulierUhlmann97_CI},
where the estimate and approximated error covariance of $\alpha$, denoted as, respectively, $\hat{\alpha}_{\textrm{c}}$ and $R_{\textrm{c}}$, are as follows:
\begin{equation}\label{Equ:FusionCI}
R_{\textrm{c}}= (\sum_i \omega_i C_i^\top R_{a_i}^{-1}C_i)^{-1},\
\hat{\alpha}_{\textrm{c}} = R_{\textrm{c}}(\sum_i \omega_i C_i^\top R_{a_i}^{-1}a_i),
\end{equation}
where $\omega_i>0,\ \forall i$ are the weights satisfy $\sum_i \omega_i = 1$.
The CI algorithm is able to preserve the consistency of the information sources (see \cite{JulierUhlmann97_CI} for detailed proofs), for any possible unknown correlated errors.
However, the optimality is traded in exchange for guaranteeing the consistency.
As a result, the CI algorithm could be too conservative especially when the information sources are less correlated.

\section{Problem formulation}

Through this paper, we use the sub-index $k\in\Zmb^+$ to denote each variable at the $k$th time instant.
Consider the following linear time-invariant (LTI) dynamic system
\begin{equation}\label{Equ:SysDynamics}
x_{k+1}=F x_{k} + B w_k,
\end{equation}
where $x_k\in\Rmb^n$ is the state of interest.
$F\in\Rmb^{n\times n}$ is the state transition matrix.
$B\in\Rmb^{n\times p}$ is a full rank matrix that models the process noise $w_k\in\Rmb^{p}$, which is assumed to be white Gaussian, i.e., $w_k\sim\mathcal {N}(0,Q)$.
Here $Q\in\Rmb^{p\times p}$ is the covariance of the process noise and it is assumed that $Q\succ \zerobf$.
Moreover, it is assumed that $\Emb[x_{k'}w_k^\top] = \zerobf$, $\forall k'\in\Zmb^*$ and $k>k'$.

Suppose that each agent $i$, $i\in\Vcal$, is able to obtain a local measurement $z_{i,k}\in\Rmb^{m_i}$, which follows a LTI sensing model:
\begin{equation}\label{Equ:LocalSenModel}
z_{i,k} = H_i x_k + v_{i,k},
\end{equation}
where $H_i\in\Rmb^{m_i\times n}$ is the local observation matrix.
$v_{i,k}$ is the local measurement noise and assumed to be white Gaussian, i.e., $v_{i,k}\sim\mathcal {N}(0,R_{i})$ with $R_{i}\in\Rmb^{m_i\times m_i}$ and $R_i\succ \zerobf$ being the covariance of the measurement noise.

In large-scale sensor networks, it is common to encounter situations where the target of interest is not directly observed by a subset of local agents.
The following assumption is made through this paper.
\begin{assumption}\label{Asump:NaiveAgentInfUncertainty}
If the target of interest is not observed by agent $i$ at time instant $k$,
then $R_i^{-1}=\zerobf$ and $z_{i,k}$ is arbitrary.
\end{assumption}

Assumption \ref{Asump:NaiveAgentInfUncertainty} essentially states that, an agent observing no target has infinite uncertainty about its measurement.
As $z_{i,k}$ can be arbitrary for such an agent, one can regard either $H_i$ or $v_{i,k}$ as arbitrary.
From this point, to simplify the statement of the observability conditions, it is assumed that $H_i=\zerobf$ if the target is not observed by agent $i$ at time $k$.
The naive agent can be therefore defined as follows.
\begin{definition}[Naive agent]\label{Def:NaiveAgent}
Let $\Nfrak\subset\Vcal$ be the set of vertices corresponding to all naive agents in the network.
Then $i\in\Nfrak$ if $(F,\textrm{col}\{H_j\}_{j\in J_i})$ is not an observable pair.
\end{definition}

In other words, an agent is naive as long as there exists one state component which cannot be recovered from the collective measurements in its inclusive neighborhood.
In such a case, the agent is actually only naive with respect to this state component.
In the rest of this paper, for simplicity, we focus on naive agents that can recover none of the state components.
In the case where naivety is discussed with respect to some state components, the same conclusions hold for the corresponding subspaces of the state space.

{
The objective is formulated as follows.
Each agent is equipped with a local sensor formulated by \eqref{Equ:LocalSenModel}.
The target makes uses of both the local information sources, and the information obtained from its neighbors through communication to come up with a consistent estimate of the dynamic state in \eqref{Equ:SysDynamics}, with as much confidence as possible.
The estimation process should not depend on any global information.}
\section{Distributed Estimation with Hybrid Information Fusion Strategy}

Suppose that at time instant $k$, each agent is able to take its local measurement $z_{i,k}$.
Besides $z_{i,k}$, each agent has a consistent prior estimate of the state of interest, denoted as $\xib$.
The local prior estimate error and approximated error covariance are defined as, respectively, $\eib\triangleq\xib-x_k$ and $\Pib\succeq \tilde{P}_{i,k|k-1}$,
where $\tilde{P}_{i,k|k-1}\triangleq\Emb[\eib\eib^\top]$ is the true error covariance of $\xib$.
The following assumptions are made through this paper.
\begin{assumption}\label{Asump:UncorrelatedViVj_ETAiVj}
For each $k\in\Zmb^+$:
(\ref{Asump:UncorrelatedViVj_ETAiVj}.1) All local measurement noises are mutually uncorrelated, i.e., $\Emb[v_{i,k} v_{j,k}^\top]=\zerobf,\ \forall i\neq j$;
(\ref{Asump:UncorrelatedViVj_ETAiVj}.2) Any local prior estimate error is uncorrelated with any local measurement noise, $\Emb[\eib v_{j,k}^\top]=\zerobf,\ \forall i,j$.
\end{assumption}


\subsection{Distributed Hybrid Information Fusion (DHIF)}
Now we discuss the strategy to fuse all information sources available to agent $i$, at $k\in\Zmb^+$.
The proposed strategy is a hybrid of two fusion steps.

\textbf{Step 1}.
Fuse the prior estimates obtained from all agents in the inclusive neighborhood, i.e., $\{\xjb\}_{j\in J_i}$, and obtain an intermediate estimate and its approximated error covariance, denoted as, respectively, $\check{x}_{i,k}$ and $\check{P}_{i,k}$.
Since the correlations between each pair of the prior estimate errors in the inclusive neighborhood are neither known nor negligible, the CI algorithm \cite{JulierUhlmann97_CI} is used for this step.
Let $\Id_n$, $\Pjb$, $\xjb$, $\forall j\in J_i$, $\check{P}_{i,k}$ and $\check{x}_{i,k}$ play, respectively, the roles of $C_i$, $R_{a_i}$, $a_i$, $\forall i\in\{1,\cdots,p\}$, $R_{\textrm{c}}$ and $\hat{\alpha}_{\textrm{c}}$ in \eqref{Equ:FusionCI}.
It follows that
$\check{P}_{i,k}^{-1} = \sum_{j\in J_i} \omega^{(k)}_{ij} \Pjb^{-1},\
\check{P}_{i,k}^{-1}\check{x}_{i,k} = \sum_{j\in J_i} \omega^{(k)}_{ij} \Pjb^{-1}\xjb,$
where $\omega^{(k)}_{ij}\in\Rmb^+$ is the weight that agent $i$ assigns to the information received from agent $j$, at time instant $k$.
The selection of $\omega^{(k)}_{ij}$ will be discussed in detail in Section \ref{Sec:ImplementDHIF:Wights}.
Note that the CI algorithm requires that $\sum_{j\in J_i}\omega^{(k)}_{ij}=1$ for any given $i$ or $k$.

\textbf{Step 2}. Fuse the intermediate estimate obtained from Step 1 (i.e., $\check{x}_{i,k}$) with all measurements obtained from the inclusive neighborhood (i.e., $\{z_j\}_{j\in J_i}$), and obtain the posterior state estimates and its approximated error covariance.
The conclusion of the following lemma is necessary for this step.
\begin{lemma}\label{Lemma:UncorrelatedEverything}
Under Assumptions (\ref{Asump:UncorrelatedViVj_ETAiVj}.1) and (\ref{Asump:UncorrelatedViVj_ETAiVj}.2), $\{v_{j,k}\}_{j\in J_i}$ and $(\check{x}_{i,k}-x_k)$ are mutually uncorrelated, $\forall i\in\Vcal$.
\end{lemma}

\begin{proof}
With Assumption (\ref{Asump:UncorrelatedViVj_ETAiVj}.1) satisfied, it suffices to show  that for each $l\in J_i$, $\Emb[(\check{x}_{i,k}-x_k) v_{l,k}^\top]=\zerobf$.
This can be verified as
\begin{equation*}
\begin{split}
\Emb[(\check{x}_{i,k}-x_k) v_{l,k}^\top]
&=\Emb[(
\check{P}_{i,k}\sum_{j\in J_i} \omega^{(k)}_{ij} \Pjb^{-1}\xjb\\
-\check{P}_{i,k}\check{P}_{i,k}^{-1}x_k) v_{l,k}^\top]
&=\check{P}_{i,k}\sum_{j\in J_i} \omega^{(k)}_{ij}\Pjb^{-1} \Emb[\ejb v_{l,k}^\top]=\zerobf,
   \end{split}
\end{equation*}
where Assumption (\ref{Asump:UncorrelatedViVj_ETAiVj}.2) is used for the last equality.
\end{proof}

As shown by Lemma \ref{Lemma:UncorrelatedEverything}, all $|J_i|+1$ information sources to fuse have mutually uncorrelated errors.
Therefore, Eq. \eqref{Equ:FusionIndep} can be applied.
Without loss of generality, let $H_j$, $R_j$ and $z_{j,k},\ \forall j\in J_i$ play, respectively, the roles of $C_i$, $R_{a_i}$ and $a_i,\ \forall i\in\{1,\cdots,p-1\}$ in \eqref{Equ:FusionIndep}, let $\Id_n$, $\check{P}_{i,k}$, $\check{x}_{i,k}$, $\Pjh$ and $\xjh$ play, respectively, the roles of $C_p$, $R_{a_p}$, $a_p$, $R_{\textrm{u}}^*$ and $\hat{\alpha}^*$ in \eqref{Equ:FusionIndep}, with $\check{P}_{i,k}$ and $\check{x}_{i,k}$ obtained from the previous step, it follows that
\begin{equation}\label{Equ:Correction}
\begin{split}
\Pih&=(\sum_{j\in J_i}\omega^{(k)}_{ij} \Pjb^{-1}+ \sum_{j\in J_i} H_j^\top R_j^{-1}H_j)^{-1},\\
\xih&= \Pih(\sum_{j\in J_i} \omega^{(k)}_{ij} \Pjb^{-1}\xjb + \sum_{j\in J_i} H_j^\top R_j^{-1}z_j).
\end{split}
\end{equation}
We refer \eqref{Equ:Correction} as the update steps of the proposed algorithm, named Distributed Hybrid Information Fusion (DHIF), implemented by agent $i$ at time instant $k$.
Here $\xih$ and $\Pih$ in \eqref{Equ:Correction} are, respectively, the local posterior estimate and the approximated local posterior error covariance by the DHIF.
Define the posterior estimate error and the true posterior error covariance as, respectively, $\eih\triangleq\xih-x_k$ and
$\tilde{P}_{i,k|k}\triangleq\Emb[\eih\eih^\top]$.
The consistency is preserved in this 2-step fusion process, as stated in the following lemma.
\begin{lemma}\label{Lemma:CorrectionStepConsistency}
Under Assumptions (\ref{Asump:UncorrelatedViVj_ETAiVj}.1) and (\ref{Asump:UncorrelatedViVj_ETAiVj}.2),
the update steps \eqref{Equ:Correction} preserve the consistency, i.e., $\Pih\succeq\tilde{P}_{i,k|k}$ if $\{z_{j,k}\}_{j\in J_i}$ and $\{\xib\}_{j\in J_i}$ are all consistent information sources.
\end{lemma}
\begin{proof}
The fact that $\xih$ given by \eqref{Equ:Correction} preserves consistency can be directly observed because both Step 1 and Step 2 are fusion strategies that preserve the consistency.
\end{proof}



\subsection{Weights Selection for CI Algorithm}\label{Sec:ImplementDHIF:Wights}
Because the CI algorithm is used in Step 1, the update steps in \eqref{Equ:Correction} only give a suboptimal estimate.
As naive agents could come up with estimates with very low confidence, if relatively high weights are assigned to such estimates, the intermediate estimate obtained by Step 1 will become less confident.
Therefore, the selection of the weights for the embedded CI algorithm in Step 1 should be determined carefully.
This is especially important in the scenario where the agents are allowed to communicate only once with their neighbors in between, as the advantages of asymptotic properties brought by the consensus-based algorithms are lost in such a scenario.
The optimal set of $\{\omega_{ij}^{(k)}\}_{j\in J_i}$ can be obtained by minimizing $\Tr ( \sum_{j\in J_i } \omega_{ij}^{(k)} \Pjb^{-1})^{-1}$,
which can be casted as the following Semi-definite Programming (SDP) problem \cite{Boyd2004convex}:
\begin{equation}\label{Equ:OptimalWeightsSDP}
\begin{split}
&\underset{u}{\textrm{minimize}}\ u^\top\onebf, \quad\quad\quad\ \\
&\textrm{subject to}
\sum_{j\in J_i} \omega_{ij}^{(k)} =1,\
0<\omega_i\leq\omega_{ij}^{(k)}\leq 1,\ \forall j\in J_i,\\
&\left[
\begin{array}{cc}
\sum_{j\in J_i} \omega^{(k)}_{ij} \Pjb^{-1} & e_l\\
e_l^\top & u_{(l)}
\end{array}\right]\succeq\zerobf,
\ l=1,\cdots,n,
\end{split}
\end{equation}
where $\underline\omega_i$, $\forall i$ are sufficiently small constant lower bounds for selecting weights.
Here $e_l\in\Rmb^n$ is the canonical basis vector whose $l$th entry is one and $u\in\Rmb^n$.
Due to the convexity of the problem formulated in \eqref{Equ:OptimalWeightsSDP}, the set of weights can be determined efficiently.
In the case where solving \eqref{Equ:OptimalWeightsSDP} is still considered as computationally expensive, some suboptimal approximations can still be used \cite{Niehsen2002FastCovIntersec}.

\subsection{Recursive Form}
\begin{algorithm}\label{Algorithm}
\If{$k=1$}
{initializes $\hat{x}_{i,1|0}$ and $P_{i,1|0}$}
computes $\Xi_{i,k}\triangleq\Pib^{-1}$ and $\xi_{i,k}\triangleq\Pib^{-1}\xib$\\
takes the local measurement $z_{i,k}$\\
computes $S_{i,k}\triangleq H_i^\top R_i^{-1}H_i$ and $ y_{i,k}\triangleq H_i^\top R_i^{-1}z_{i,k}$\\
sends $S_{i,k}$, $y_{i,k}$, $\Xi_{i,k}$ and $\xi_{i,k}$ to agent $j$, $\forall j$ such that $i\in N_j$\\
receives $S_{j,k}$, $y_{j,k}$, $\Xi_{j,k}$ and $\xi_{j,k}$ from agent $j$, $\forall j\in N_i$\label{Alg:Step:InputMultiple}\\
selects the set of weights $\{\omega_{ij}^{(k)}\}_{j\in J_i}$\label{Alg:Step:ChooseWeights}\\
computes $\bar{\Xi}_{i,k}=\sum_{j\in J_i} \omega^{(k)}_{ij} \Xi_{i,k}$ and $
\bar{\xi}_{i,k}=\sum_{j\in J_i} \omega^{(k)}_{ij} \xi_{i,k}$\label{Alg:Step:CI}\\
computes $\bar{y}_{i,k}=\sum_{j\in J_i}{y}_{j,k}$ and $ \bar{S}_{i,k}=\sum_{j\in J_i}{S}_{j,k}$\label{Alg:Step:KF}\\
updates local estimate and approximated covariance
\begin{eqnarray}
\label{Equ:Alg:CovUpdate}
\Pih&=&\left(\bar{S}_{i,k} + \bar{\Xi}_{i,k}\right)^{-1}\\
\label{Equ:Alg:StateUpdate}
\xih&=& \Pih\left(\bar{y}_{i,k} + \bar{\xi}_{i,k} \right)
\end{eqnarray} \\\label{Alg:Update}
predicts local estimate and approximated covariance
\begin{eqnarray}
\label{Equ:Alg:CovPrd}
P_{i,k+1|k}&=&F \Pih F^\top + B Q B^\top\\
\label{Equ:Alg:StatePrd}
\hat{x}_{i,k+1|k}&=&F \xih
\end{eqnarray}
\caption{DHIF Implemented by Agent $i$ at Time $k$}
\end{algorithm}

With the prediction steps formulated based on the system dynamics in \eqref{Equ:SysDynamics}, the recursive form of the distributed hybrid information fusion (DHIF) strategy implemented by agent $i$ at time $k$ is summarized in Algorithm \ref{Algorithm}.

\begin{theorem}\label{Thm:Consistency}
Under Assumption \ref{Asump:UncorrelatedViVj_ETAiVj}, Algorithm \ref{Algorithm} preserves the consistency.
That is,
$\Pib\succeq\tilde{P}_{i,k|k-1}$ and $\Pih\succeq\tilde{P}_{i,k|k}$, $\forall k\in\Zmb^+$,
if the initialized local estimates $\{\hat{x}_{i,1|0}\}_{i=1}^{N}$ are consistent. That is,
 \begin{equation}\label{Equ:InitialCondition}
   P_{i,1|0}\succeq\Emb[(\hat{x}_{i,1|0}-x_1)(\hat{x}_{i,1|0}-x_1)^\top],\ \forall i.
 \end{equation}
\end{theorem}
\begin{proof}
Lemma \ref{Lemma:CorrectionStepConsistency} has shown that the update steps in \eqref{Equ:Correction} preserve the consistency of the information.
Therefore, it suffices to show if $\Pih\succeq\tilde{P}_{i,k|k}$, one can obtain $P_{i,k+1|k}\succeq\tilde{P}_{i,k+1|k}$.
Note that $\tilde{P}_{i,k+1|k}
=F\tilde{P}_{i,k|k}F^\top + BQB^\top + \Omega_1 + \Omega_2$,
where $\Omega_1 = F\Emb[\eih w_{k+1}^\top]B^\top$ and $\Omega_2=\Omega_1^\top$.
Since $\eih$ in a linear combination of $\left\{ x_0,\{w_l\}^{k}_{l=1},\{v_l\}^{k}_{l=1} \right\}$, each of which is assumed to be uncorrelated with $w_{k+1}$, it follows that $\Emb[\eih w_{k+1}^\top]=\zerobf$.
Therefore, $\Omega_2=\Omega_1^\top=\zerobf$.
It follows from \eqref{Equ:Alg:CovUpdate} that if $\Pih\succeq\tilde{P}_{i,k|k}$,
$P_{i,k+1|k}=F \Pih F^\top + B Q B^\top\succeq F\tilde{P}_{i,k|k}F^\top + BQB^\top=\tilde{P}_{i,k+1|k}.$
Therefore, the prediction steps in Algorithm \ref{Algorithm} also preserve the consistency.
It follows that any estimate at any agent obtained by the recursive Algorithm \ref{Algorithm} is consistent,
if all local estimates initially fed to the algorithm are consistent, as formulated in \eqref{Equ:InitialCondition}.
\end{proof}
\begin{remark}
It is worth mentioning that \eqref{Equ:InitialCondition} can be easily satisfied in general.
The prior knowledge about the state of interest can be learned in an off-line manner before the fusion process.
In the worst case, each agent can simply choose $P_{i,1|0}^{-1}=\zerobf$, which indicates the infinite initial local uncertainty so that \eqref{Equ:InitialCondition} is satisfied.
\end{remark}


\section{Comparisons with Existing Algorithms}

In this section, we analyze the advantages of the proposed DHIF algorithm, and compare it with some existing algorithms in the literature.

\subsection{Robustness in Presence of Naive Agents}
Note that \eqref{Equ:Alg:StateUpdate}
can be written as
\begin{equation}\label{Equ:Update:ConsensusForm}
\begin{split}
\xih&=\xib
+ \Pih\left(\bar{y}_{i,k}-\bar{S}_{i,k}\xib\right)\\
&+\Pih\sum_{j\in J_i} \omega^{(k)}_{ij}\Pjb^{-1} (\xjb-\xib).
\end{split}
\end{equation}
As observed in \eqref{Equ:Update:ConsensusForm}, the approximated posterior estimate error covariance $\Pih$, which implies the uncertainty of agent $i$'s local estimate, is multiplied to both of the last two terms.
On one hand, if agent $i$ is more confident about its own local estimate, $\Pih$ will be small.
In such a situation, the effects on $\xih$ caused by the last two terms in \eqref{Equ:Update:ConsensusForm} will be attenuated by a small $\Pih$.
On the other hand, when $\Pih$ is large,
\eqref{Equ:Update:ConsensusForm} will push the local posterior estimate to match its neighbors' measurements and their local prior estimates,
as described from the following two perspectives.

\textbf{(1)}
Note that $\bar{y}_{i,k}-\bar{S}_{i,k}\xib=\sum_{j\in J_i}H_j^\top R_j^{-1}(z_{i,k}-H_j\xib)$,
which can be regarded as a weighted sum of the innovation terms that push $\xih$ to match $z_{j,k}$, $\forall j\in J_i$.
If a neighboring agent $j$ does not directly observe the state of interest, $R_j^{-1}=\zerobf$.
Thus, its bad effects on the update of $\xih$, caused by the erroneous local measurement $z_{j,k}$, will be eliminated in an automated manner.

\textbf{(2)}
Similarly, the last term in \eqref{Equ:Update:ConsensusForm} can also be regarded as a weighted sum, which pushes $\xih$ to match $\xjb$, $\forall j\in J_i$.
The weight that agent $i$ assigns to agent $j$ is the multiplication of $\omega^{(k)}_{ij}$ and $\Pjb^{-1}$.
The former component is selected to increase the confidence of the estimate while preserving its consistency;
the latter one, which implies the confidence of each neighbor's local prior estimate, pushes the local estimate toward its neighbors' with higher confidence.

Let $\delta$ be a scalar parameter. If the last term in \eqref{Equ:Update:ConsensusForm} is replaced with $\delta\Pih\sum_{j\in J_i} (\xjb-\xib)$,
the update step of the local posterior estimates proposed by the KCF is obtained.
The KCF\cite{olfati2009kalman} has been shown to perform well with guaranteed convergence under certain conditions.
However, it should be noted that the KCF relies on the assumption that the state of interest is jointly observable in the inclusive neighborhood of every agent.
The divergence of the KCF in presence of naive agents has been shown in \cite{kamal2011GKCF} and \cite{kamal2013ICFTAC}.
Essentially with the last term of \eqref{Equ:Update:ConsensusForm} replaced with $\delta\Pih\sum_{j\in J_i} (\xjb-\xib)$ in the KCF,
each agent equally weighs its neighbors' prior estimates regardless of the confidence of these estimates.
The performance is therefore deteriorated especially
when there exist some other naive agents in its neighborhood.
More detailed analysis on the performance of the KCF in presence of naive agents can be found in \cite{kamal2011GKCF}.
\subsection{Fully Distributed}

In the GKCF\cite{kamal2011GKCF} and ICF\cite{kamal2013ICFTAC}, the rate parameter $\epsilon$ is required to be chosen between $0$ and $1/\Delta_\text{max}$ to guarantee the performance of the embedded average consensus algorithm.
If the maximum in-degree is changing with time,
a proper selection of $\epsilon$ might not be as good as before.
Even if $\Delta_\text{max}$ is known, it is not clear how to select a nice $\epsilon$ in the studies of the GKCF and ICF.
In \cite{kamal2011GKCF,kamal2013ICFTAC}, after obtaining a new local measurement, each agent is allowed to communicate with its local neighbors for infinite times, before it finally updates its posterior local estimate at the current time instant.
Therefore, in such a case, any selection between $0$ and $1/\Delta_{\text{max}}$ will guarantee that every agent's local estimate asymptotically becomes identical before its next measurement update.
Unfortunately, in the realistic case where every agent communicates only once with its neighbors before updating its local estimate, the importance of selecting $\epsilon$ is similar to the selection of weights in the proposed DHIF, as discussed in Section \ref{Sec:ImplementDHIF:Wights}.

Besides $\Delta_\text{max}$, the ICF also requires each agent to know the total number of agents $N$ in the network to asymptotically approach the centralized solution via infinite communication steps before the agents update their local estimates.
Similar to the discussion on $\Delta_{\text{max}}$, $N$ could be changing over time.
Moreover, in the case focused in this paper (single communication step), the ICF is not able to obtain the optimal solution even with the correct knowledge of $N$.

The proposed DHIF algorithm does not require any global information. 
It is run in an automated manner and is adaptive to the locally unknown changes in the network.

\subsection{Consistency of Estimates}\label{Sec:Advantages:Consistency}

The consistency is one of the most fundamental but significant properties to be preserved during the information fusion process.
That is, the approximated error covariance should be lower bounded by the true error covariance.
The approximated error covariance of an inconsistent estimator is over-confident, and hence cannot indicate the uncertainty of the estimate.
Let the counterpart of $\Pih$ and $\Pib$ obtained by the ICF be, respectively, $\Pih^\textrm{\tiny ICF}$ and $\Pib^\textrm{\tiny ICF}$.
In the case where the agents communicate only once before updating their local estimates,
the update steps of the ICF proposed in  \cite{kamal2013ICFTAC} can be written as
\begin{equation}\label{Equ:ICF:Updates}
\begin{split}
\Pih^\textrm{\tiny ICF}
=[\sum_{j\in J_i}\sigma_{ij}\cdot {(\Pjb^\textrm{\tiny ICF})^{-1}}+  N\sum_{j\in J_i}\sigma_{ij} H_j^\top R_j^{-1}H_j]^{-1}, 
\end{split}
\end{equation}
where $\sigma_{ij},\ \forall j$ are the weights
related to the rate parameter $\epsilon$.
Therefore, it is very possible that
\begin{equation}\label{Equ:ICF:AsymConstency}
  N\sum_{j\in J_i}\sigma_{ij} H_j^\top R_j^{-1}H_j\succeq \sum_{j\in J_i}H_j^\top R_j^{-1}H_j.
\end{equation}
Note that the right-hand side of the above equation is the total information contained in $\{z_{j,k}\}_{j\in J_i}$.
Therefore,  it is possible that the $\Pih^\textrm{\tiny ICF}$ obtained from \eqref{Equ:ICF:Updates} is smaller than the true error covariance of the local posterior estimates, especially when $N\gg|J_i|$, which is usually the case in sparse wireless sensor networks.


\subsection{Confidence of Estimates}\label{Sec:Advantages:Confidence}

While the consistency is guaranteed, one would also come up with an estimate that is as confident as possible.
In the case where the agents communicate only once before local updates, the update step of the approximated error covariance obtained by the KLA algorithm \cite{Battistelli2014KLA_Automatica}, denoted as $\Pih^\textrm{\tiny KLA}$, has the following form:
\begin{equation}\label{Equ:KLA:Updates}
\begin{split}
\Pih^\textrm{\tiny KLA}&=\{\sum_{j\in J_i}\sigma_{ij} [(\Pjb^\textrm{\tiny KLA})^{-1}+  H_j^\top R_j^{-1}H_j]\}^{-1},
\end{split}
\end{equation}
where $\sigma_{ij}=\{\Sigma\}_{(i,j)}$ with $\Sigma$ being a stochastic matrix associated with the communication graph.
In general, given the same set of prior estimates and local measurements from the inclusive neighborhood, the local posterior estimate obtained by \eqref{Equ:Correction} is more confident than that obtained by \eqref{Equ:KLA:Updates}.
This is mainly due to the following two reasons.

\textbf{(1)}
The KLA algorithm did not specify the selection of the weights $\sigma_{ij},\ \forall j$.
The only constraint formulated in \cite{Battistelli2014KLA_Automatica} is $\Sigma$ being a primitive stochastic matrix.
When only one communication step is feasible,
the weights really matter in minimizing the fused $\Pih^\textrm{\tiny KLA}$.

\textbf{(2)}
More importantly, even if we suppose that the same set of weights, say $\sigma_{ij},\ \forall j\in J_i$, are selected for both fusion processes in \eqref{Equ:KLA:Updates} and \eqref{Equ:Correction}, the update steps of the proposed DHIF are still guaranteed to give more confident estimates since $\sum_{j\in J_i}H_j^\top R_j^{-1} H_j\succeq\sum_{j\in J_i}\sigma_{ij} H_j^\top R_j^{-1} H_j$, which in turn, implies that $\Pih^\textrm{\tiny KLA}\succeq\Pih$.

Superficially,  \eqref{Equ:Correction} might look similar to \eqref{Equ:KLA:Updates}.
However, the philosophies behind the proposed DHIF and the KLA are very different.
Compared to the proposed DHIF, the update steps of the KLA can also be regarded as a two-step fusion process.
The first step is that, for each $i\in\Vcal$, fuse $z_{i,k}$ and $\xib^\textrm{\tiny KLA}$ into an intermediate estimate $\check{x}^\textrm{\tiny KLA}_{i,k}$ with $\check{P}^\textrm{\tiny KLA}_{i,k}$ being its approximated error covariance.
The fusion strategy \eqref{Equ:FusionIndep} is adopted as $z_{i,k}$ and $\xib^\textrm{\tiny KLA}$ have uncorrelated errors.
The second step is to fuse $\{\check{x}^\textrm{\tiny KLA}_{j,k}\}_{ j\in J_i}$.
Note that the information sources in this set have unknown correlations in general due to the fact that $\{\xjb^\textrm{\tiny KLA}\}_{ j\in J_i}$ become highly correlated as the fusion process goes on.
Therefore, the consensus on the PDFs, which is equivalent to the CI algorithm in the case therein (see \cite{Battistelli2014KLA_Automatica} for details), is used to obtain the posterior estimate and approximated error covariance.
Note that during the overall fusion process of the KLA algorithm, the fact that all local measurement noise are mutually uncorrelated is not utilized.
By fusing the $\xib^\textrm{\tiny KLA}$ and $z_{i,k}$ first, $z_{i,k}$ is treated as correlated with $z_{j,k}, \forall j\in N_i$ with unknown correlations.
This makes the final estimate become relatively conservative.

\section{Boundedness of the Estimate Error}

In this section, the sufficient condition for each local estimate error to be upper bounded, is formulated.
Specifically, we assume the communication topology and the set of naive agents are both fixed.
The scenario in which they are time-varying is left for the future work.
Due to the similarities in terms of the structures of the DHIF and KLA algorithms, we borrow the framework of the proofs in \cite{Battistelli2014KLA_Automatica}, and add our extensions.
The following lemma will be used in the proof of the main result.

\begin{lemma}\label{Lemma:SpanForestVsPostiveEntries}
Let $D$ be the stochastic matrix associated with $\Gcal(\Vcal,\Ecal)$.
Let $\Rcal$ and $\Bcal$ be two exclusive sets of vertices such that $\Rcal\cup\Bcal=\Vcal$.
Suppose that $\Gcal(\Rcal,\Ecal_\Rcal)$ is strongly connected,
where $\Ecal_\Rcal\triangleq\{(i,j)|(i,j)\in\Ecal,i,j\in\Rcal\}$.
Also suppose that there exists a directed path from $\Rcal$ to every agent $j\in \Bcal$ but there is no edge from any agent in $\Bcal$ to $\Rcal$.
Then there exists a $\bar{k}$ such that $\{D^{k}\}_{(i,j)}>0$, $\forall i\in\Vcal$ and $j\in\Rcal$, if $k\geq\bar{k}$.
Moreover, the same conclusion holds for $\{\Dfrak_k\}_{(i,j)}$, where
$\Dfrak_k$ is the multiplication of $k$ matrices which belong to the stochastic matrices family induced by $D$.
\end{lemma}

\begin{proof}
Without loss of generality, let $\Rcal=\{1,\cdots,|\Rcal|\}$ and $\Bcal = \{|\Rcal|+1,\cdots,N\}$.
Let $D_\Rcal\in\Rmb^{|\Rcal| \times| \Rcal|}$ be the stochastic matrix associated with the strongly connected graph $\Gcal(\Rcal,\Ecal_\Rcal)$.
Thus $D_\Rcal$ is primitive.
Because there is no edge from $\Bcal$ to $\Rcal$,
$D$ can be written as a block lower triangular matrix
\begin{eqnarray*}
D = \left[
\begin{array}{cc}
D_\Rcal & \zerobf\\
D_{\Bcal\Rcal} & D_\Bcal
\end{array}
\right],
\end{eqnarray*}
where $D_{\Bcal\Rcal}$ and $D_\Bcal$ are matrices with appropriate dimensions.
It follows that
\begin{eqnarray*}
D^{{k}} = \left[
\begin{array}{cc}
D_\Rcal^{{k}} & \zerobf\\
D'_{\Bcal\Rcal}  & D_\Bcal^{{k}}
\end{array}
\right],
\text{ where }
D'_{\Bcal\Rcal}=\sum_{r=0}^{{k}-1}D_\Bcal^r D_{\Bcal\Rcal}D_\Rcal^{{k}-1-r}.
\end{eqnarray*}
As $D_\Rcal$ is primitive, $D_\Rcal^{{k}} > \zerobf$ for any sufficiently large $k$.
Thus it remains to show that $D'_{\Bcal\Rcal}>\zerobf$ for any sufficiently large $k$.
This can be shown by induction.

Suppose that when $k=k'$, $D_\Rcal^k>\zerobf$.
Also suppose that agent $ i_1\in\Bcal$ has a direct edge from $\Rcal$.
Thus, $\{D\}_{( i_1,j\in\Rcal)}$ contains at least one positive entry.
It follows that\begin{equation*}
\begin{split}
\{D^{k'+1}\}_{( i_1,j\in\Rcal)}
&=
\{D D^{k'}\}_{( i_1,j\in\Rcal)}
=
\{D \}_{( i_1,:)}
\{ D^{k'}\}_{(:,j\in\Rcal)}\\
&\geq
\{D \}_{( i_1,j\in\Rcal)}D_\Rcal^{k'},
\end{split}
\end{equation*}
where the last inequality is due to the fact that $D\geq\zerobf$.
As $\{D\}_{( i_1,j\in\Bcal)}$ contains at least one positive entry and $D_\Rcal^{k'}>\zerobf$,
it follows that $\{D^{k'+1}\}_{( i_1,j\in\Rcal)}>\zerobf$.
Moreover, it follows that as $k'\rightarrow\infty$,
$D_\Rcal^{k'}\rightarrow
\onebf v_1^\top$,
where $ v_1>\zerobf$ is the left eigenvector associated with the simple eigenvalue $1$ of $D_\Rcal$ satisfying $\| v_1\|_1=1$.
That is, each entry in $\{D^k \}_{( i_1,j\in\Rcal)}$ is uniformly lower bounded above, for any $k\geq k'$.
Let $ i_2\in\Bcal$ be a child vertex of $ i_1$.
Thus, $\{D\}_{( i_2, i_1)}>0$.
Since $\{D^{k'+1} \}_{( i_1,j\in\Rcal)}>\zerobf$ as previously shown, it follows that
\begin{equation*}
\begin{split}
&\{D^{k'+2}\}_{( i_2,j\in\Rcal)}
=
\{D D^{k'+1}\}_{( i_2,j\in\Rcal)}\\
=&
\{D \}_{( i_2,:)}
\{ D^{k'+1}\}_{(:,j\in\Rcal)}
\geq
\{D \}_{( i_2, i_1)}\{D_\Rcal^{k'+1}\}_{( i_1, j\in \Rcal)}>\zerobf.
\end{split}
\end{equation*}
By the similar approach, it can also be shown that $\{D^k \}_{( i_2,j\in \Rcal)}>\zerobf$ for any $k\geq k'+1$.
Therefore, it can be eventually shown that for any agent $i$ that has a directed path from some leader component $\Rcal$, there exists a sufficiently large $\bar k$ such that,
 for any $k\geq\bar k$, $\{D^{k}\}_{(i,j)}>0$, $\forall j\in\Rcal$.

As each corresponding entry of all matrices, which belong to the stochastic matrix family induced by the same matrix $D$, has the same type (zero/nonzero), the same conclusion holds for $\{\Dfrak_k\}_{(i,j)}$
\end{proof}

The following lemma is also used for the proof later on.
\begin{lemma} [\cite{Battistelli2014KLA_Automatica}]{\label{Lemma:GB_Lemma_upper_bounded_PSD}}
Let $F$ be a nonsingular matrix.
Then for any $Y\succ\zerobf$ and $\check{\Omega}\succeq\zerobf$, there exists a $\check{\beta}\in(0,1]$ such that
$(F\Omega^{-1} F^\top+Y)^{-1}\succeq\check{\beta}F^{-\top}\Omega F^{-1}$ for any $\Omega^{-1}\succeq\check{\Omega}^{-1}$.
\end{lemma}
\begin{theorem}\label{Thm:UpperBoundedCov}
Suppose that Assumptions \ref{Asump:NaiveAgentInfUncertainty}, \ref{Asump:UncorrelatedViVj_ETAiVj} and Eq. \eqref{Equ:InitialCondition} are satisfied.
Also suppose that the state transition matrix $F$ is nonsingular.
Then for each agent $i\in\Vcal$,
there exists $\bar{k}\in\Zmb^+$ and $\bar{P}_i\succ\zerobf\in\Rmb^{n\times n}$, such that
$\bar{P}_i\succeq\Pih$, $ \forall k>\bar{k}$,
if there exists at least one strongly connected subgraph $\Gcal(\Rcal,\Ecal_\Rcal)$ with joint observability, where $\Rcal\subseteq\Vcal$ and $\Ecal_\Rcal\triangleq\{(i,j)|(i,j)\in\Ecal,i,j\in\Rcal\}$
that has a directed path from $\Rcal$ to agent $i$.
\end{theorem}
\begin{proof}
As the proof follows the framework of Theorem 3 in \cite{Battistelli2014KLA_Automatica}, only the part contributed by this paper will be shown in detail.
As shown in Theorem \ref{Thm:Consistency}, if \eqref{Equ:InitialCondition} is satisfied, then for any $k\in\Zmb^+$,
$P_{i,k|k}\succeq \tilde{P}_{i,k|k}\succ \zerobf$.
Therefore by Lemma \ref{Lemma:GB_Lemma_upper_bounded_PSD},
$(F\Pih F^\top+BQB^\top)^{-1}\succeq\check{\beta}F^{-\top}\Pih^{-1} F^{-1}$.
Define $D_k\in\Rmb^{n\times n}$ such that $\{D_k\}_{(i,j)}\triangleq d_{ij}^{(k)}$, where $d_{ij}^{(k)}=\omega_{ij}^{(k)}$ if $j\in J_i$ and $d_{ij}^{(k)}=0$ if $j\notin J_i$, $\forall i,j$.
It follows from \eqref{Equ:Alg:CovUpdate} that
\begin{equation}\label{Equ:RecursiveCov}
  \begin{split}
\Pih^{-1}&=\bar{S}_{i,k} + \sum_{j} d^{(k)}_{ij}(FP_{j,k-1|k-1}F^\top+BQB^\top)^{-1}\\
&\succeq\bar{S}_{i,k} + \sum_{j} d^{(k)}_{ij}\check{\beta}F^{-\top} P_{j,k-1|k-1}^{-1} F^{-1},
   \end{split}
\end{equation}
Define
$\Dfrak^q_p\triangleq D_q D_{q-1}\cdots \bar D_p$ for some $q> p$, where $\bar {\{ \cdot \}}$ is an element-wised ceil function.
Note that $\Dfrak^q_p$ has the same structure with $D_k^{q-p+1}$ for any $k\in\Zmb^+$.
Suppose that $\bar{k}\geq 2$, using the recursive inequality in \eqref{Equ:RecursiveCov}, it follows that $\Pih^{-1}\succeq\ \check{\beta}^{\bar{k}}\sum_j \{\Dfrak_{k-\bar{k}+1}^k\}_{(i,j)}F^{-\bar{k}\top}P_{j,k-\bar{k},k-\bar{k}}F^{-\bar{k}}+\Gamma_i $, where
\begin{equation*}
\begin{split}
\Gamma_i =& \bar{S}_{i,k} + \sum_{\tau=2}^{\bar{k}}  \check{\beta}^{\tau-1}\sum_j \{\Dfrak_{k-\tau+2}^k\}_{(i,j)}F^{(1-\tau)\top}H_j^\top R_j^{-1} H_j F^{1-\tau}\\
\succeq& \sum_{\tau=2}^{\bar{k}}  \check{\beta}^{\tau-1}\sum_{j\notin \Nfrak} \{\Dfrak_{k-\tau+2}^k\}_{(i,j)}F^{(1-\tau)\top}H_j^\top R_j^{-1} H_j F^{1-\tau}\triangleq\Gamma'_i.
\end{split}
\end{equation*}
The last inequality is due to the fact that $\bar{S}_{i,k}\succeq\zerobf$ and $R_j^{-1}=\zerobf,\ \forall j\in \Nfrak$.
For simplicity, suppose that $\Gcal(\Rcal,\Ecal_\Rcal)$ has a directed path to every other agent\footnote{
When there are multiple disjoint strongly connected subgraphs,
the subgraphs can be treated separately.
}.
Thus Lemma \ref{Lemma:SpanForestVsPostiveEntries} can be used\footnote
{Note that although the first condition in Theorem \ref{Thm:UpperBoundedCov} does not assume there is no edge path from agent $i\notin\Rcal$ to agent $j\in\Rcal$,
this can always be guaranteed by including all $i\notin\Rcal$ which has a directed path to $j\in\Rcal$ into $\Rcal$.
Therefore, Lemma \ref{Lemma:SpanForestVsPostiveEntries} can be applied.}.
Therefore, there exists $\bar{\tau}$ such that if $\tau\geq\bar{\tau}$,
$\{\Dfrak_{k-\tau+2}^k\}_{(i,j)}>0, \forall i\in\Vcal$ and $ j\in\Rcal$.
Recall that $\check{\beta}>0$ and $R_j\succ\zerobf,\ \forall j\notin\Nfrak$.
It follows that if agents in $\Rcal$ has joint observability,
$\Gamma'_i\succ\zerobf$ for any $k\geq\bar{k}\triangleq\bar{\tau}+n$.
Let $\bar{P}_i\triangleq(\Gamma'_i)^{-1}$.
It follows that $\bar{P}_i\succ\zerobf$ and $\bar{P}_i\succeq \Pih$.
\end{proof}

Theorem \ref{Thm:UpperBoundedCov} formulates the sufficient condition for the bounded local posterior estimate errors.
Essentially, the local estimate errors at a certain agent is bounded as long as the agent is able to eventually obtain the knowledge about the entire state from somewhere.


The strategies to decrease the upper bound of $\Pih$ for each $i\in\Vcal$ can also be observed from the proof of Theorem \ref{Thm:UpperBoundedCov}.
Note that $\Pih^{-1}\succeq\Gamma'_i\succ\zerobf$ with $\Gamma'_i$ defined in the proof of sufficiency.
This implies that
the upper bound of $\Pih$ can be decreased if $\Gamma'_i$ is increased.
This is achieved if:
1) there are fewer naive agents so that the summation with respect to $j$ has more positive terms;
2) the graph is denser so that for a certain $\bar{k}$ the summation with respect to $\tau$ has more positive terms; 3) local measurements are more accurate (smaller $R_j$) or redundant ($H_j$ with more rows) so that each term in the summation has a greater value.

\section{Simulations}
In the simulation, the target follows the linear dynamic system in \eqref{Equ:SysDynamics}.
The parameters are adopted from the simulation example used in \cite{Battistelli2014KLA_Automatica}, and listed as follows:
\begin{equation*}
F =
\left[
\begin{array}{cc}
\Id_2 & \Delta_T\Id_2\\
\zerobf & \Id_2
\end{array}\right],\
Q =
\left[
\begin{array}{cc}
5\Delta_T^3 \Id_2/3 & 5\Delta_T^2 \Id_2/2 \\
5\Delta_T^2 \Id_2/2  & 5\Delta_T\Id_2
\end{array}\right],\
B = \Id_4,
\end{equation*}
where $\Delta_T=4$ is the sampling interval.
\begin{figure}
  \centering
  \includegraphics[width=.5\linewidth]{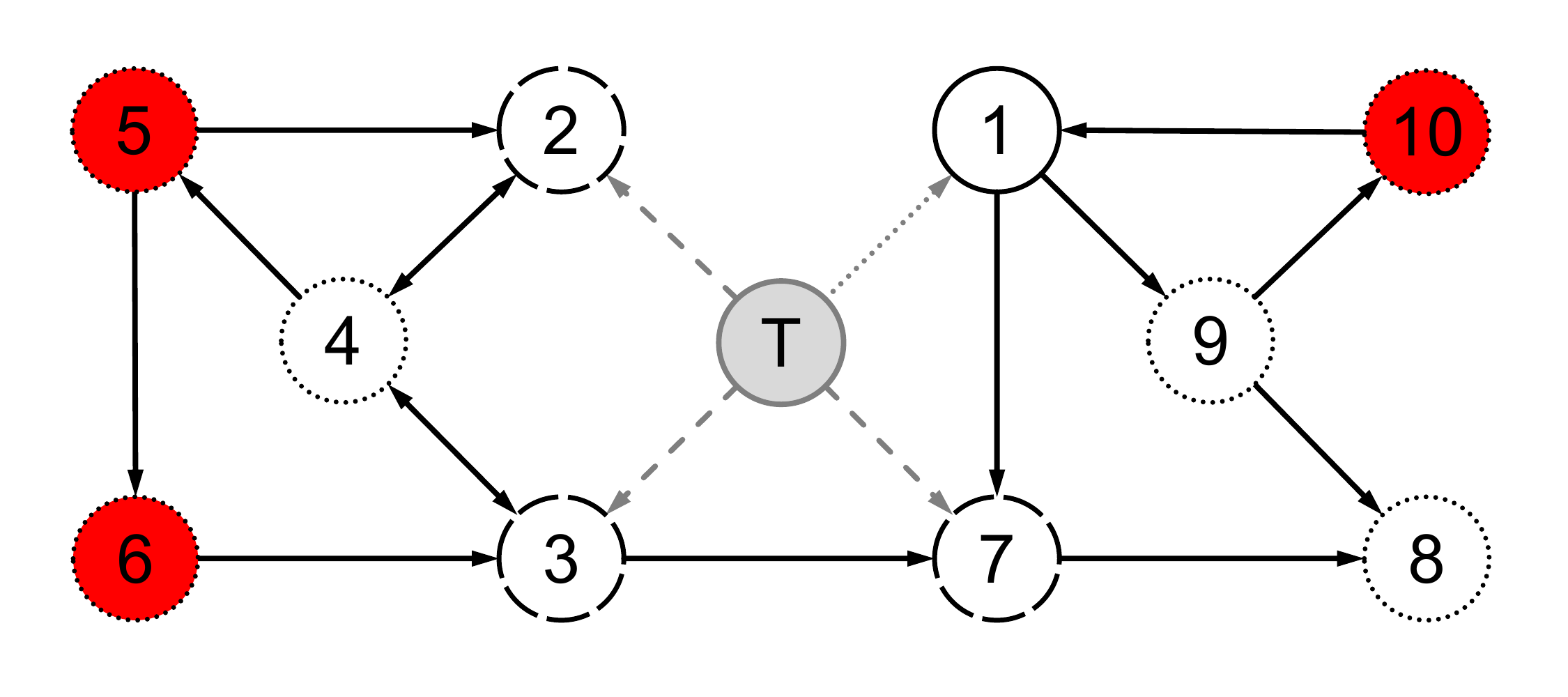}
  \caption{Communication topology for simulation.
  T: target;
  Dashed circle: agents with partial observability;
  Solid circle: agents with full observability;
  Dotted circle: agents not observing the target;
  Red: naive agent.\label{Fig:Simulation:Topo} }
\end{figure}
A distributed sensor network with $N=10$ agents, whose directed communication topology is shown in Figure \ref{Fig:Simulation:Topo}, is used to track the state of the target.
Let $H_2=[1\ 0\ 0\ 0]$, $H_3=H_7=[0\ 1\ 0\ 0]$ and $H_1=[H_2^\top\ H_3^\top]^\top$.
Correspondingly, let $R_2=R_3=R_7=225$, $R_1=225\Id_2$ .
Also let $H_i=\zerobf$ and $R_i^{-1}=\zerobf$ for the other agents.
Note that there exists no directed spanning tree in this topology.
However, the conditions formulated in Theorem \ref{Thm:UpperBoundedCov} are satisfied.

The proposed DHIF algorithm is implemented, where the weights in Step \ref{Alg:Step:ChooseWeights} of Algorithm \ref{Algorithm} are solved from \eqref{Equ:OptimalWeightsSDP}.
The KLA algorithm and the ICF are selected for comparisons.
As the selection of weights was not specifically mentioned in \cite{Battistelli2014KLA_Automatica}, we let $\sigma_{ij}=1/|J_i|,\forall i$.
The rate parameter $\epsilon$ used in the ICF is selected to be the same as that in \cite{kamal2013ICFTAC}, i.e., $\epsilon=0.65/\Delta_{\textrm{max}}$.
A hypothetical centralized Kalman filter, whose observation matrix has the form of $H\triangleq \textrm{col}\{H_i\}_{i=1}^N$, is used as the benchmark.

\begin{figure}
\centering
\subfigure[Agent 1]{\label{eihVS3sigma:1} 
\includegraphics[width=.85\linewidth]{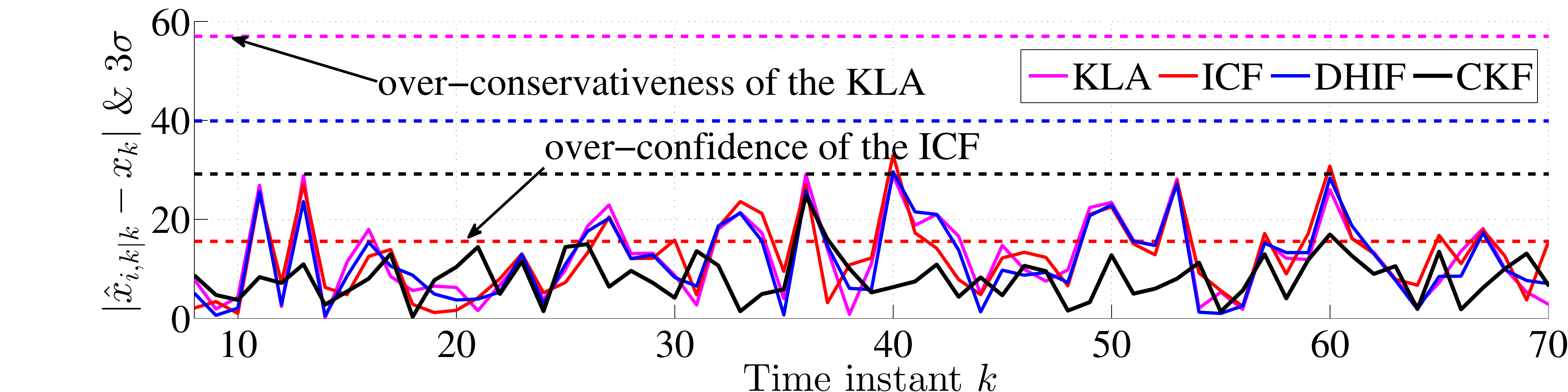}
}
\subfigure[Agent 6]{\label{eihVS3sigma:6} 
\includegraphics[width=.85\linewidth]{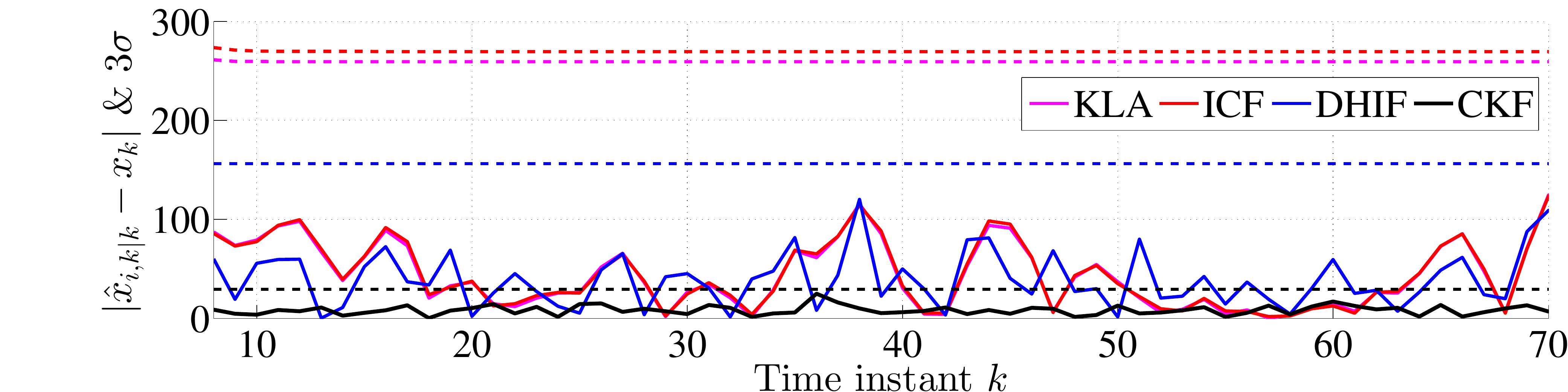}
}
\caption{$|\xih-x_k|$ and $3\sigma$-bound given by each algorithm: x-position}
\end{figure}
The local absolute posterior estimate errors of agent 1 (locally observable) and agent 6 (very naive), obtained by each algorithm, are plotted, respectively, in Figure \ref{eihVS3sigma:1} and Figure \ref{eihVS3sigma:6}.
The $3\sigma$-bound corresponding to each algorithm, with $\sigma$ being the standard deviation (STD) computed from the approximated error covariance by the algorithm, is also plotted in the same color but dashed line.
As observed in Figure \ref{eihVS3sigma:1}, the posterior estimate error of the ICF's exceeds its $3\sigma$-bound frequently.
Actually, its $3\sigma$-bound is even less than the one obtained by the CKF, which provides the minimum possible MSE.
Therefore, the ICF can be overconfident when the agents communicate with each other for only once before updating their local estimates, as analyzed in Section \ref{Sec:Advantages:Consistency}.
In both Figure \ref{eihVS3sigma:1} and Figure \ref{eihVS3sigma:6}, the $\sigma$ value computed by the KLA is around 1.5 times than that computed by the DHIF.

\begin{figure}
\centering
\includegraphics[width=.95\linewidth]{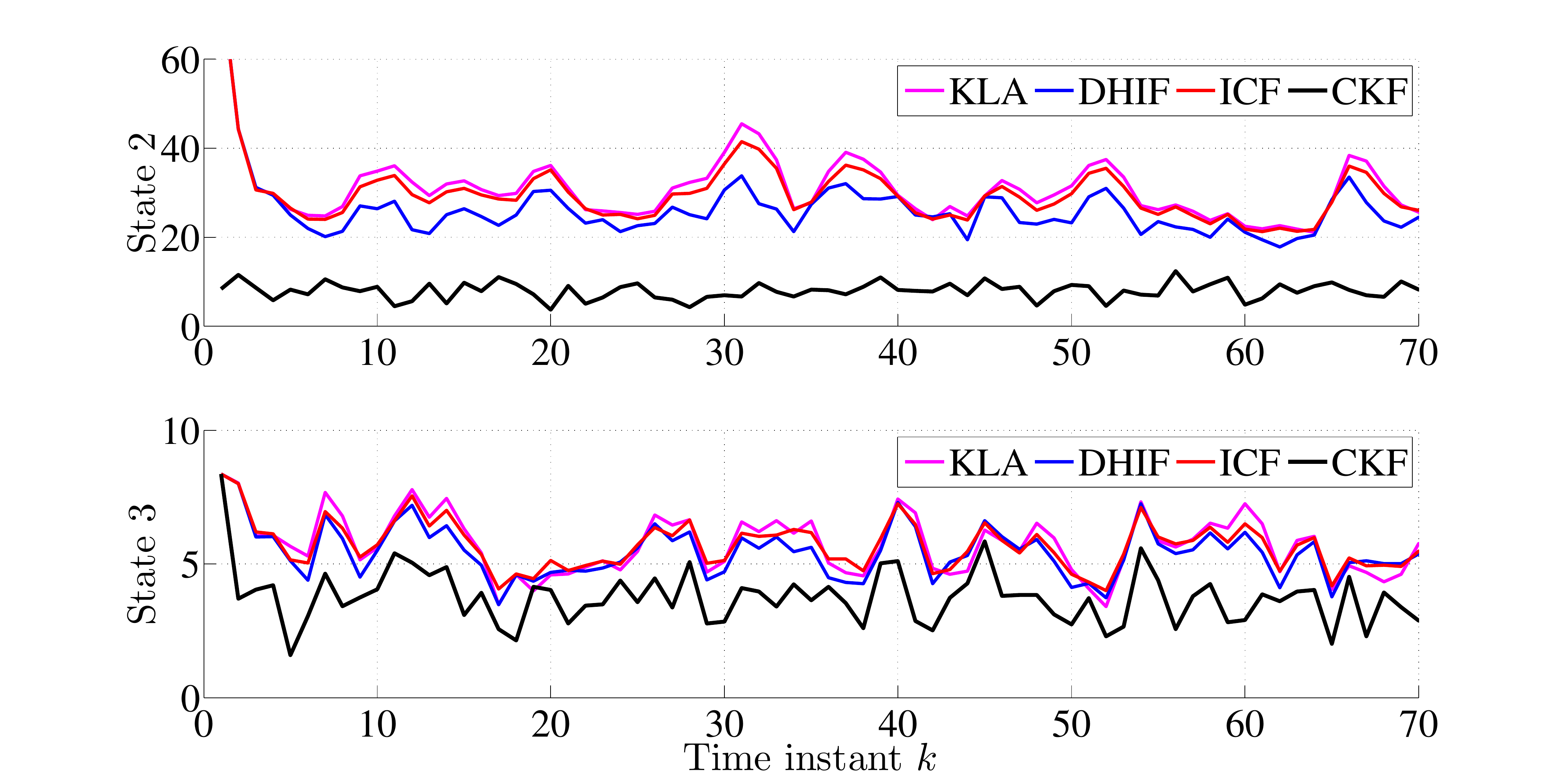}
\caption{Comparisons of RMSE $\psi_k=\sqrt{\frac{1}{T}\sum_{j=1}^T\frac{1}{N}\sum_{i=1}^N(\xih^{(j)}-x_k^{(j)})^2}$ obtained by each algorithm, where the superscript ``$(j)$" is the trial index\label{fig:RMSE}}
\end{figure}
The same simulation as the previous one are further implemented for 500 trials.
The rooted mean squared errors (RMSE) in estimating two state components, denoted as $\psi_k$, for $k\in[1,70]$, are plotted in Figure \ref{fig:RMSE}.
As observed, the proposed DHIF has the lowest RMSE among all distributed algorithms.


\section{Conclusions}
The problem of distributed state estimation of a linear dynamic process in the wireless sensor networks has been considered.
A distributed hybrid information fusion (DHIF) algorithm has been proposed.
In the proposed algorithm, each agent communicates with its neighbors for only once, before updating its local estimates and approximated error covariance.
In the proposed algorithm, the consistency of estimate is guaranteed to be preserved.
Meanwhile, the confidence of each local estimate is improved by efficiently utilizing all information sources available to the local agent.
The proposed algorithm is fully distributed and robust against the presence of naive agents.
The sufficient condition has been formulated to guarantee that the local estimate error being bounded at steady state.


\addtolength{\textheight}{-12cm}   

\bibliographystyle{IEEEtran}
\bibliography{refs}
\end{document}